\pdfoutput=1

\documentclass[elsarticle.cls,superscriptaddress, notitlepage]{revtex4-1}
\usepackage{amsmath}
\usepackage{latexsym}
\usepackage{amssymb}
\usepackage{graphics,epstopdf}
\usepackage{graphicx}
\usepackage[colorlinks=true, citecolor=blue, urlcolor=blue ]{hyperref}
\usepackage{float}
\usepackage{graphicx}
\usepackage{color}
\usepackage{amsfonts}
\usepackage{soul}
\usepackage{amsmath,amsthm}
\newtheorem{theorem}{Theorem}
\newtheorem{corollary}{Corollary}
\usepackage{amsmath}
\DeclareMathOperator{\Tr}{Tr}

\begin{document}

\title{Testing local-realism and macro-realism under generalized dichotomic measurements}

\author{Debarshi Das}
\email{debarshidas@jcbose.ac.in}
\affiliation{Centre for Astroparticle Physics and Space Science (CAPSS), Bose Institute, Block EN, Sector V, Salt Lake, Kolkata 700 091, India}

\author{Shiladitya Mal}
\email{shiladitya.27@gmail.com}
\affiliation{S. N. Bose National Centre for Basic Sciences, Block JD, Sector III, Salt Lake, Kolkata  700 098, India}

\author{Dipankar Home}
\email{dhome@jcbose.ac.in}
\affiliation{Centre for Astroparticle Physics and Space Science (CAPSS), Bose Institute, Block EN, Sector V, Salt Lake, Kolkata 700 091, India}

\begin{abstract}
Generalized quantum measurements with two outcomes are fully characterized by two real parameters, dubbed as sharpness parameter and biasedness parameter and they can be linked with different aspects of the experimental setup. It is known that sharpness parameter characterizes precision of the measurements and decreasing sharpness parameter of the measurements reduces the possibility of probing quantum features like quantum mechanical (QM) violation of local-realism (LR) or macro-realism (MR). Here we investigate the effect of biasedness together with that of sharpness of measurements and find a trade-off between those two parameters in the context of probing QM violations of LR and MR. Interestingly, we also find that the above mentioned trade-off is more robust in the latter case.
\end{abstract} 


\maketitle

\section{Introduction}
Nonclassical features of quantum mechanics, for example, quantum mechanical (QM) violations of local-realism (LR) \cite{bell'65, chsh} or macro-realism (MR) \cite{leggett'85} are probed through performing incompatible measurements on systems. Ideal measurements, also known as projective measurements, are described by a set of projectors acting on the system's Hilbert space. This is also known as Von Neumann measurement after his seminal work formalising measurement scheme in QM \cite{von}. Later this concept is extended to  positive operator valued measurement (POVM) and presently it describes the most general kind of quantum measurements. POVM has operational advantages in many tasks over projective measurements, for example, distinguishing nonorthogonal states \cite{disting}, demonstrating hidden nonlocality \cite{popescu, gisin}.

In POVM formalism two observables can be measured jointly even when they do not commute. It is well known \cite{fine} that the observables which can be measured jointly do not lead to the violations of  Bell-CHSH (Bell-Clauser-Horne-Shimony-Holt) inequalities \cite{bell'65, chsh}. Moreover for two dichotomic measurements, POVMs are not better than projective measurements in the context of QM violations of Bell-CHSH inequalities \cite{clave}. It was shown by considering the implication of quantum entanglement to nonlocal game, which is a kind of cooperative game of incomplete information. Nonlocal game can be described as follows: a referee, who determines the game, randomly chooses questions, drawn from  finite sets according to some fixed probability distribution and send them to two players (say, Alice and Bob) at distant locations.  Alice and Bob respond to the referee with an answer without communicating themselves. The referee then evaluates some predicate based on the questions asked and their answers, to determine whether they win or lose.  Alice and Bob can gain advantage in winning if they share quantum correlations instead of classical correlations \cite{clave}.

Condition for joint measurability of two noncommuting observables was derived \cite{busch'86} and for two dichotomic observables it is fully characterized  \cite{busch2, stano, yu}. It is also shown that for particular two-level observables the border of joint measurability coincides with the one for the violation of the Bell-CHSH inequality \cite{and}. In \cite{wolf} it has been shown that, for two non-jointly measurable observables at one side, there always exists state and projective measurements for the other side such that the violation of the CHSH inequality is enabled. This result was shown in  \cite{wolf} by casting joint measurability, considering its implicit characterization, as a problem of semi-definite programme (SDP) \cite{sdp}.

Moving to the practical origin of POVM, it is known that they occur in quantum measurement formalism mainly due to two reasons \cite{busch3}. Firstly, POVM may account for the ever-present imperfections of any measurement and secondly, there are measurement situations for which there exists no ideal background observable represented by a projective valued measurement (PVM). Examples of the second reason include genuine phase space observables, as the individual measurement outcomes are fuzzy phase space points in accordance with the Heisenberg uncertainty relation. Regarding the first reason it is known that in Von Neumann measurement scheme there is a cut between classical and quantum domain where quantum systems are measured by classical apparatus. System variable to be measured becomes entangled with classical probe in the process of measurement interaction. By sharply distinguishing different probe states one can achieve PVM. In reality, measurements are usually
 not PVM reflecting non-zero overlap between the probe states. For detail study one can see \cite{busch3}.  
 
 Dichotomic POVM are characterised by two real parameters dubbed as sharpness and biasedness \cite{busch2, yu}. These two parameters can be linked with different types of nonidealness, with respect to ideal projective measurement, involved in the real experimental scenario. Sharpness characterizes measurement precision which is related to the overlap between non-orthgonal probe states \cite{mal1}. On the other hand, biasedness can be linked with the error in alignment of Stern-Gerlach (SG) apparatus or the deviation from the Gaussian nature of spatial wave-packet of incident spin-$\frac{1}{2}$ particles which is recently shown by some of us \cite{mal}. Any physical system is described with respect to some reference frame. For example, spin direction is defined with respect to some gyroscope in the laboratory instead of any purported absolute Newtonian space. Setting up SG apparatus requires a attached reference frame with respect to which the direction of inhomogeneous magnetic field, incident particle beam, position of screen are defined. Incident particle beam may not pass through the center of SG apparatus due to some alignment problem which reflects in the biasedness of POVM measured with such non-ideal device. In information processing tasks communication between different observers without shared reference frame is an interesting area of research and for detailed study one can see \cite{rev}. These two quantities, therefore, have well defined physical interpretations beyond mathematical constructions.
 
It is known that decreasing the sharpness of measurements reduces the possibility of obtaining  QM violations of Bell-CHSH inequality \cite{busch4, guru} or Leggett-Garg inequality (LGI) \cite{olgi, saha, das} and below a certain value of the sharpness parameter, violations of these inequalities are not obtained. In the most general formulation of dichotomic measurements, as we have just mentioned, there is another parameter, apart from sharpness parameter, which is known as biasedness parameter.

In this paper we explore the effect of biasedness of measurements on probing quantumness in the context of QM violation of CHSH inequality as well as in the context of QM violations of three inequivalent necessary conditions for MR, namely LGI \cite{leggett'85}, Wigner's form of the Leggett-Garg inequality (WLGI) \cite{saha} and the condition of no-signalling in time (NSIT) \cite{nsit}. Inequivalence of these necessary conditions of MR has been studied \cite{swati1} and it was recently shown \cite {swati2} that for a particular biased unsharp measurement there exists a state of two level system for which all these necessary conditions of MR are violated for any non-zero value of the sharpness parameter. In another work \cite{bum}  the effect of biasedness over that of unsharpness of multi-outcome spin measurements has been explored for multilevel spin systems considering a particular measurement scheme.

 In case of spatial correlations we find out the effect of the biasedness parameter on the minimum value of the sharpness parameter over which the QM violation of CHSH inequality persists. Furthermore, we derive the necessary and sufficient condition for the violation of CHSH inequality with biased unsharp measurements at one side and projective measurements at another side. As a corollary of this derivation we find out that the violation of the CHSH inequality cannot be enabled with dichotomic POVMs at one side if that is not enabled with projective measurements on both sides. This result is consistent with previous findings \cite{clave}. In case of temporal correlations we find out the effect of the biasedness parameter on the minimum values of the sharpness parameter over which QM violations of different necessary conditions of MR persist. Thus it is shown that there is a trade-off between the sharpness parameter and the biasedness parameter characterizing arbitrary dichotomic POVM in the context of probing QM violations of local-realist (LR) and macro-realist (MR) inequalities. It is also observed that the above mentioned trade-off is more robust in the latter case which means that the effect of biasedness parameter counters the effect of unsharpness of measurements more in the latter case.
 
We organize this paper in the following way. We briefly discuss  the characterization of the most general dichotomic POVM in Section II. In Section III, we consider the QM violation of the CHSH inequality with most general dichotomic POVM at one side. Then in Section IV we show the trade-off between sharpness and biasedness parameter in probing the QM violations of three inequivalent necessary conditions of MR, i.e., LGI, WLGI, NSIT. Section V contains discussion and concluding remarks.

\section{Generalized dichotomic measurements}
Projective valued measurement (PVM) is a set of projectors that add to identity, i.e., $A\equiv \lbrace P_{i}\vert\sum P_{i}=\mathbb{I}, P_i^2 = P_i \rbrace$ (where $P_i$s are
projectors). The probability of getting the $i$-th outcome is given by, $\Tr[ \rho P_i ]$ for the state $\rho$.

On the other hand, positive operator valued measurement (POVM) is a set of positive operators that add to identity, i.e., $E\equiv \lbrace E_{i}\vert\sum E_{i}=\mathbb{I},0< E_i\leq \mathbb{I}\rbrace$.  The probability of getting the $i$-th outcome is $\Tr[\rho E_{i}]$. Effects ($E_{i}s$) represent quantum events that may occur as outcomes of a measurement.

 In case of dichotomic measurements, the most general POVM is characterized by two parameters - sharpness parameter ($\lambda$) and biasedness parameter ($\gamma$) and the corresponding effect operators are given by,
 \begin{equation}
 \label{ub}
 E^{\pm} = \lambda P^{\pm} + \frac{1 \pm \gamma - \lambda}{2} \mathbb{I},
 \end{equation}
where $P^{+}$ and $P^{-}$ are sharp projectors corresponding to the two outcomes $+$ and $-$ respectively. For $E^{\pm}$ being valid effect operators, the positivity ($E^{\pm} \geq 0$) and normalisation condition ($E^{+} + E^{-} = \mathbb{I}$) have to be satisfied. From these conditions it is followed that $|\lambda|+|\gamma|\leq 1$. Sharpness parameter ($\lambda$) characterizes the measurement precision which is related to the overlap between non-orthgonal probe states \cite{mal1}. We consider $\lambda$ to be positive as negative value of the sharpness parameter has no physical meaning. Eq.(\ref{ub}) with $\gamma = 0$ gives unbiased unsharp measurement, which is also a dichotomic POVM, but not the most general one \cite{busch3}. For unbiased unsharp measurement $(1 - \lambda)$ characterizes the amount of unsharpness associated with the measurement.

\section{QM violation of local-realism with most generalized dichotomic measurements}
Quantum mechanical predictions are incompatible with local realist theory, which is probed through QM violation of Bell-CHSH inequality. Let us consider two spatially seperated parties, say Alice and Bob. Alice performs two dichotomic observables $A$ and $A^{'}$; Bob performs two dichotomic observables $B$ and $B^{'}$. In this scenario the CHSH inequality \cite{chsh} is given by
 \begin{equation}\label{chs}
 \langle A B\rangle + \langle A B^{'}\rangle+\langle A^{'} B\rangle -\langle A^{'} B^{'}\rangle\leq 2.
 \end{equation}
 $\langle A B\rangle$ is the correlation between measurements of dichotomic observables $A$ and $B$ performed by Alice and Bob respectively. Other terms in Inequality (\ref{chs}) are similarly defined.
 
There exist states and observables such that the maximum value of the left hand side (LHS) of CHSH inequality (\ref{chs}) is given by $2\sqrt{2}$ and in quantum theory $2\sqrt{2}$ is the maximum possible value of the CHSH expression which is known as Cirelson's bound \cite{cirelson}. It is known that for unbiased unsharp measurements at one side ($0 \leq \lambda \leq 1$ and $\gamma=0$) and projective measurements at another side the CHSH inequality is not violated for $\lambda \leq \frac{1}{\sqrt{2}}$ \cite{busch4, guru}.

Now to explore the role of biasedness over unsharpness in the context of QM violation of CHSH inequality, we consider that Alice performs biased
unsharp measurements (\ref{ub}) whereas Bob's
measurements are projective. We consider the question how the minimum value of the sharpness parameter, above which the QM violation of CHSH inequality persists, is modified by the biasedness parameter $\gamma$ of Alice's biased unsharp measurement. For singlet state we find that the CHSH inequality is not violated for $\lambda \leq \frac{1}{\sqrt{2}}$ whatever be the value of $\gamma$.

Now we consider arbitrary two qubit state under biased unsharp measurements at Alice's side. As biased unsharp measurement is the most general
POVM for two outcome scenario, this actually establishes
an effective criteria whether a given state violates CHSH
inequality under the consideration of POVMs at one side. In \cite{Horodecki'1995} Horodecki family established the necessary and sufficient criteria
for the QM violation of CHSH inequality by any two-qubit state under projective measurements.
In a similar spirit, we derive the necessary and sufficient criteria for the QM violation of CHSH inequality by any two-qubit system under biased unsharp measurements performed by Alice.

Any arbitrary state in $\mathcal{H}(=\mathbb{C}^2\otimes\mathbb{C}^2)$ can be expressed in terms of Hilbert-Schmidt basis as
\begin{eqnarray}\label{rho}
\rho=\frac{1}{4}(\mathbb{I}\otimes\mathbb{I}+\vec{r} \cdot \vec{\sigma}\otimes\mathbb{I}+\mathbb{I}\otimes\vec{s} \cdot \vec{\sigma}+\sum_{i,j=1}^{3}t_{ij}\sigma_i\otimes\sigma_j).
\end{eqnarray}
Here $\mathbb{I}$ is identity operator acting on $\mathbb{C}^2$, $\sigma_i$s are three Pauli matrices and $\vec{r}, \vec{s}$ are vectors in $\mathbb{R}^3$ with norm less than equal to unity. $\vec{r} \cdot \vec{\sigma} = \sum_{i=1}^{3} r_i \sigma_i$. $\vec{s} \cdot \vec{\sigma} = \sum_{i=1}^{3} s_i \sigma_i$. The coefficients $t_{ij}$ = $\Tr ( \rho \sigma_i \otimes \sigma_j )$ form a real matrix which we shall denote by $T$. $T$ is called the correlation matrix.  In addition, for being a valid density matrix, $\rho$ has to be normalized and positive semi-definite.

Let us define a matrix $V=TT^t$, where $T$ is the correlation matrix of the state (\ref{rho}) with coeffiecient $t_{ij}$ = $\Tr( \rho \sigma_i \otimes \sigma_j )$, $T^t$ represents transposition of $T$. The $3 \times 3$ real matrix $V$ is a symmetric one, so it can be diagonalized. Let $v, \tilde{v}$ are the two greatest, obviously positive, eigenvalues of the matrix $V$ \cite{Horodecki'1995}. We define a quantity, 
\begin{eqnarray}
M(\rho)=v+\tilde{v}.
\end{eqnarray}

\begin{theorem}
There exist biased unsharp measurements with sharpness parameter $\lambda$ and biasedness parameter $\gamma$ at Alice's side and projective measurements at Bob's side for which any two qubit density matrix (\ref{rho}) violates the CHSH inequality  iff $\lambda\sqrt{M(\rho)}+ | \gamma | |\vec{s}|>1$,  where $\vec{s} = (s_1, s_2, s_3)$ is a vector in $\mathbb{R}^3$ with norm less than equal to unity; $s_i =  \Tr( \rho \mathbb{I} \otimes \sigma_i )$ ($i = 1, 2, 3$).
\end{theorem}

\begin{proof}
Consider that the measurements performed at both sides are projective. In this case let us assume that $Q=\hat{q}.\vec{\sigma}$, where $Q\in\{A, A^{\prime}, B, B^{\prime}\}$ and  $\hat{q}\in\{\hat{a}, \hat{a^{\prime}}, \hat{b}, \hat{b^{\prime}}\}$; $\hat{a}, \hat{a^{\prime}}, \hat{b}, \hat{b^{\prime}}$ are unit vectors in $\mathbb{R}^3$. Using these the Bell operator corresponding to the CHSH inequality (\ref{chs}) can be expressed as,
\begin{eqnarray}
\mathcal{B}_{CHSH}^{\lambda = 1,\gamma = 0}= \hat{a}.\vec{\sigma}\otimes( \hat{b}+\hat{b}^{\prime}).\vec{\sigma}+\hat{a}^{\prime}.\vec{\sigma}\otimes( \hat{b}-\hat{b}^{\prime}).\vec{\sigma}.
\end{eqnarray}
Now, with biased unsharp measurements (with sharpness parameter $\lambda$ and biasedness parameter $\gamma$) at Alice's side, the expectation value of the Bell operator corresponding to the CHSH inequality (\ref{chs}) becomes
\begin{eqnarray}
\label{chshbu}
\langle\mathcal{B}_{CHSH}^{\lambda, \gamma }\rangle=\lambda \langle\mathcal{B}_{CHSH}^{\lambda = 1,\gamma = 0}\rangle +2 \gamma \hat{b}.\vec{s},
\end{eqnarray}
where $\langle\mathcal{B}_{CHSH}^{\lambda = 1,\gamma = 0}\rangle$ is the expectation value of the Bell operator corresponding to the CHSH inequality (\ref{chs}) when both parties perform projective measurements and it is given by \cite{Horodecki'1995},
\begin{equation}
\langle\mathcal{B}_{CHSH}^{\lambda = 1,\gamma = 0}\rangle =  ( \hat{a},T(\hat{b} + \hat{b^{'}})) +  ( \hat{a},T(\hat{b} - \hat{b^{'}})).
\end{equation}
Here $( \hat{a},T(\hat{b} + \hat{b^{'}}))$ and $( \hat{a},T(\hat{b} - \hat{b^{'}}))$ present Euclidean scalar products in $\mathbb{R}^3$.
We have to maximize the quantity given by Eq.(\ref{chshbu}) over all measurement settings. Following the prescription described in \cite{Horodecki'1995}, let us take $\hat{b}+\hat{b}^{\prime}=2\cos\theta\hat{c}$ and $\hat{b}-\hat{b}^{\prime}=2\sin\theta\hat{c}^{\prime}$. Where $\hat{c}, \hat{c}^{\prime}$ are mutually orthogonal unit vectors  in $\mathbb{R}^3$; $\theta \in [0, \pi/2]$. With this choice the maximum expectation value $\langle\mathcal{B}_{CHSH}^{\lambda, \gamma }\rangle_{max}$ over all possible measurement settings is given by,

\begin{eqnarray}
\langle\mathcal{B}_{CHSH}^{\lambda,\gamma}\rangle_{max} = \mathrm{max}_{\hat{a},\hat{a}^{\prime},\hat{c},\hat{c}^{\prime},\theta} \big[ 2\lambda \{ (\hat{a},T\hat{c})\cos\theta+(\hat{a}^{\prime},T\hat{c}^{\prime})\sin\theta \}
+2\gamma(\cos\theta\hat{c}+\sin\theta\hat{c}^{\prime}).\vec{s} \big].
\label{neee}
\end{eqnarray}

In order to maximize over $\hat{a},\hat{a}^{\prime}$, we choose these to be in the direction of $T\hat{c}, T\hat{c}^{\prime}$ respectively. Hence, from Eq.(\ref{neee}) we obtain

\begin{align}
\langle\mathcal{B}_{CHSH}^{\lambda,\gamma}\rangle_{max} & =  \mathrm{max}_{\hat{c},\hat{c}^{\prime},\theta}[2\lambda(||T\hat{c}||\cos\theta +||T\hat{c}^{\prime}||\sin\theta)+
2(\cos\theta \gamma \hat{c}.\vec{s}+\sin\theta \gamma \hat{c}^{\prime}.\vec{s})] \nonumber\\
&=\mathrm{max}_{\hat{c},\hat{c}^{\prime}}[2\lambda\sqrt{||T\hat{c}||^2+||T\hat{c}^{\prime}||^2}+2 \sqrt{|\gamma \hat{c}.\vec{s}|^2+| \gamma \hat{c}^{\prime}.\vec{s}|^2}] \nonumber\\
&=2\lambda\sqrt{M(\rho)}+2 |\gamma | |\vec{s}|.
\label{nsp}
\end{align}

Here, $||T\hat{c}||$ presents Euclidean norm in $\mathbb{R}^3$; $||T\hat{c}||^2 = ( \hat{c},T^{t} T\hat{c})$. We have used $\max \limits_{\hat{c},\hat{c}^{\prime}} \large( ||T\hat{c}||^2+||T\hat{c}^{\prime}||^2 \large)$ = $v+\tilde{v}$ = $M(\rho)$ following \cite{Horodecki'1995}. Hence, in the present context the bipartite qubit state (\ref{rho}) violates the CHSH inequality iff the maximum value of the LHS of CHSH inequality is greater than $2$, i. e., 
\begin{eqnarray}
A(\rho, \lambda, \gamma) = \lambda\sqrt{M(\rho)}+ |\gamma | |\vec{s}|>1.
\label{nsm}
\end{eqnarray}
This is the necessary and sufficient criteria determining whether a given bipartite qubit state violates CHSH inequality when  Alice performs the most generalized dichotomic measurements, i.e., biased unsharp measurements and Bob performs projective measurements.
\end{proof}

The above theorem implies the following corollary:

\begin{corollary}
If a bipartite qubit state $\rho$ does not violate CHSH inequality for projective measurements on both sides, then the state will not violate CHSH inequality with POVMs at one side.
\end{corollary}

\begin{proof}
For $E^{\pm}$ in Eq.(\ref{ub}) being valid effect operators, $\lambda$ and $\gamma$ must satisfy $|\lambda|+|\gamma|\leq 1$. Moreover, we consider sharpness parameter $\lambda$ to be positive as negative value of the sharpness parameter has no physical meaning \cite{mal1}. Let us assume that the arbitrary bipartite qubit state $\rho$ given by Eq.(\ref{rho}) does not violate the CHSH inequality with projective measurements at both sides, i. e., $M(\rho) \leq 1$. Now consider that Alice performs biased unsharp measurements with sharpness parameter $\lambda$ and biasedness parameter $\gamma$, Bob performs projective measurements. In this scenario there are the following four possible cases:\\ 

\emph{Case 1. $\lambda \geq 0$, $\gamma >0, \lambda +\gamma =1$}: In this case\\ 
i) Consider that $ \sqrt{M(\rho)}-|\vec{s}|=\epsilon \geq 0$. With this $A(\rho, \lambda, \gamma)$ becomes $\sqrt{M(\rho)}-\epsilon (1-\lambda)\leq \sqrt{M(\rho)} \leq 1$.\\
ii) Consider that $\sqrt{M(\rho)}<|\vec{s}|$. Hence, $A(\rho, \lambda, \gamma)$=  $|\vec{s}|-\lambda(|\vec{s}|-\sqrt{M(\rho)}) \leq |\vec{s}| \leq $ $1$. \\

\emph{Case 2. $\lambda \geq 0$, $\gamma >0, \lambda +\gamma <1$}: In this case\\ 
i) Consider that $ \sqrt{M(\rho)}-|\vec{s}|=\epsilon \geq 0$. With this $A(\rho, \lambda, \gamma)$ becomes $(\lambda +\gamma)\sqrt{M(\rho)}-\gamma\epsilon<\sqrt{M(\rho)} \leq 1$. \\
ii) Consider that $\sqrt{M(\rho)}<|\vec{s}|$, $ |\vec{s}| - \sqrt{M(\rho)} = \epsilon > 0$. Hence,  $A(\rho, \lambda, \gamma)=(\lambda +\gamma)|\vec{s}|-\lambda\epsilon < |\vec{s}| \leq$ $1$. \\

\emph{Case 3. $\lambda \geq 0$, $\gamma \leq 0, \lambda - \gamma =1$}: In this case\\ 
i) Consider that $ \sqrt{M(\rho)}-|\vec{s}|=\epsilon \geq 0$. With this $A(\rho, \lambda, \gamma)$ becomes $\sqrt{M(\rho)}-\epsilon (1-\lambda)\leq \sqrt{M(\rho)} \leq 1$.\\
ii) Consider that $\sqrt{M(\rho)}<|\vec{s}|$, $ |\vec{s}| - \sqrt{M(\rho)} = \epsilon > 0$. Hence,  $A(\rho, \lambda, \gamma)$=  $|\vec{s}|-\lambda \epsilon \leq |\vec{s}| \leq $ $1$. \\

\emph{Case 4.  $\lambda \geq 0$,  $\gamma \leq 0$, $\lambda -\gamma <1$}: In this case\\ 
i) Consider that $ \sqrt{M(\rho)}-|\vec{s}|=\epsilon \geq 0$. With this $A(\rho, \lambda, \gamma)$ becomes $(\lambda - \gamma)\sqrt{M(\rho)}-|\gamma | \epsilon<\sqrt{M(\rho)} \leq 1$. \\
ii) Consider that $\sqrt{M(\rho)}<|\vec{s}|$, $ |\vec{s}| - \sqrt{M(\rho)} = \epsilon > 0$. Hence,  $A(\rho, \lambda, \gamma)=(\lambda - \gamma)|\vec{s}|-\lambda\epsilon < |\vec{s}| \leq$ $1$. \\

Hence, we find that in all possible four cases $A(\rho, \lambda, \gamma) \leq 1$ if $M(\rho) \leq 1$. It can, therefore, be concluded that, for any two qubit state, there is no violation of CHSH inequality with POVMs at one side if that is not enabled with PVMs at both sides.

One important point to be stressed here is that we have considered only positive values of  sharpness parameter $\lambda$ as negative value of the sharpness parameter has no physical meaning \cite{mal1}. However, if one is interested to consider negative values of $\lambda$ also, then it can be checked that the above result remains unchanged.

\end{proof}

This result is consistent with what was obtained in \cite{clave}, where it was shown that projective measurements are sufficient for obtaining maximal violation of CHSH inequality considering optimal strategy of binary game and using the relation between binary XOR game (which is one particular nonlocal game mentioned in the introduction) with the CHSH value. On a related note it should be mentioned that, in the context of self-testing of binary observables considering Bell-CHSH violation, it has been shown \cite{jed} that when the maximal violation of the CHSH inequality is obtained, the certified local observables indicate projective measurements performed by the parties.

From Eq.(\ref{nsp}) it is found that the maximum value of the LHS of CHSH inequality (\ref{chs}) for the state (\ref{rho}) when both parties perform projective measurements ($\lambda = 1$, $\gamma = 0$) is given by $2 \sqrt{M(\rho)}$ and that when Alice performs unbiased unsharp measurements ($0 \leq \lambda \leq 1$, $\gamma=0$) is given by $2 \lambda \sqrt{M(\rho)}$. Hence, unsharpness of Alice's measurements (quantified by ``$1-\lambda$") decreases the maximum value of the LHS of CHSH inequality (\ref{chs}). The maximum value of the LHS of CHSH inequality (\ref{chs}) for the state (\ref{rho}) when Alice performs biased unsharp measurements is given by $2 \lambda \sqrt{M(\rho)} + 2 | \gamma | | \vec{s} |$. Hence, for a given unsharpness (i. e., for a given $\lambda$) of Alice's measurements, the maximum value of the LHS of CHSH inequality (\ref{chs}) for the state (\ref{rho}) is increased by incorporating any nonzero biasedness $\gamma$. Now, unsharpness and biasedness of a measurement characterizes different types of non-idealness involved in the measurement in real experimental scenario \cite{mal1, mal}. Hence, it can be concluded that non-idealness linked with unsharpness of the measurements reduces the quantum effect in the form of QM violation of CHSH inequality (\ref{chs}) by any two-qubit system. On the other hand, introducing non-idealness as characterized by biasedness in the measurements enables revelation of quantum effects even with a value of sharpness parameter for which such quantum effects in the unbiased scenario is not observed. Thus the trade-off arises.

\section{QM violation of macro-realism with most generalized dichotomic measurements}

In order to probe macroscopic coherence Leggett and Garg introduces an inequality (LGI) \cite{leggett'85} which is a consequence of conjunction of two assumptions, namely, (A1) {\it{Macrorealism per se}} and (A2) {\it{Non-Invasive Measurability}} (NIM). (A1) implies a macroscopic object, which has available to it two or more macroscopically distinct states, is at any given time in a definite one of those states. NIM states that it is possible, in principle, to determine which of the states the system is in, without affecting the state itself or the system's subsequent evolution. We denote these two assumptions together as macro-realism (MR). One can replace (A1) by a stronger assumption of {\it{realism}} which states that at any instant, irrespective of measurement, a system is in any one of the available definite states such that all its observable properties have definite values. Any model satisfying realism and NIM is known as non-invasive realist model which satisfies MR. Laws of classical physics satisfy these assumptions, whereas QM description of nature is incompatible with MR. In the present work we consider three inequivalent necessary conditions of MR namely LGI, WLGI \cite{saha} and NSIT \cite{nsit} to explore the role of biasedness of measurements over unsharpness of measurements. As our motivation here is to explore the role of biasedness of measurements, we consider microscopic system, say, qubit system rather than system with large dimension. 

To test falsifiability of the assumptions of MR sequential measurements are performed at different times on a system evolving under some Hamiltonian. From the measurement statistics two time correlation function is defined as $C_{ij}=\langle Q(t_i)Q(t_j)\rangle=\sum_{m_i,m_j=\pm}m_i m_j$ $p(Q(t_i) =m_i, Q(t_j) = m_j)$, where $p(Q(t_i) =m_i, Q(t_j) = m_j)$ is the joint probability of getting the outcomes $m_i$ and $m_j$ when the two-outcome ($\pm$) observable $Q$ is measured at instances $t_i$ and $t_j$ respectively. In this scenario the simplest form of Leggett-Garg inequality (LGI), which is a necessary condition of MR, is given by,
\begin{equation}\label{lgi}
K_{LGI}=C_{12}+C_{23}-C_{13}\leq 1.
\end{equation} 
Now consider precession of a spin-$1/2$ particle with initial state $\rho = \frac{1}{2} ( \mathbb{I} + \vec{r} \cdot \vec{\sigma})$  (where $\vec{r} \cdot \vec{\sigma} = r \sin \theta \cos \phi \sigma_x + r \sin \theta \sin \phi \sigma_y + r \cos \theta \sigma_z$; $0 \leq \theta \leq \pi$; $0 \leq \phi < 2 \pi$; $r$ is real and $0 \leq r \leq 1$) at instance $t_1$ under 
the unitary evolution $U_{t}=e^{-i\omega t\sigma_{x}/2}$, where $\omega$ is the 
angular precession frequency. 
Let us choose equidistant measurement times with time difference $\Delta t(=t_{2}-t_{1} = t_3 - t_2)$. For projective measurements the LHS of LGI (\ref{lgi}) is independent of the state parameters $r$, $\theta$, $\phi$. The maximum value taken by the LHS of LGI (\ref{lgi}) with projective measurements ($\lambda = 1$, $\gamma=0$) corresponding to the operator $\sigma_z$ at instances $t_1$, $t_2$ and $t_3$ is $\frac{3}{2}$. This happens for $\omega\Delta t = \frac{\pi}{6}$.

If the measurements corresponding to the operator $\sigma_z$ are unbiased unsharp ($0 \leq \lambda \leq 1$ and $\gamma=0$), then it has been shown that LGI given by Inequality (\ref{lgi}) is not violated for $\lambda \leq \sqrt{2/3}$ \cite{saha}.
\begin{figure}[t!]
    \begin{center}
    \resizebox{9cm}{9cm}{\includegraphics{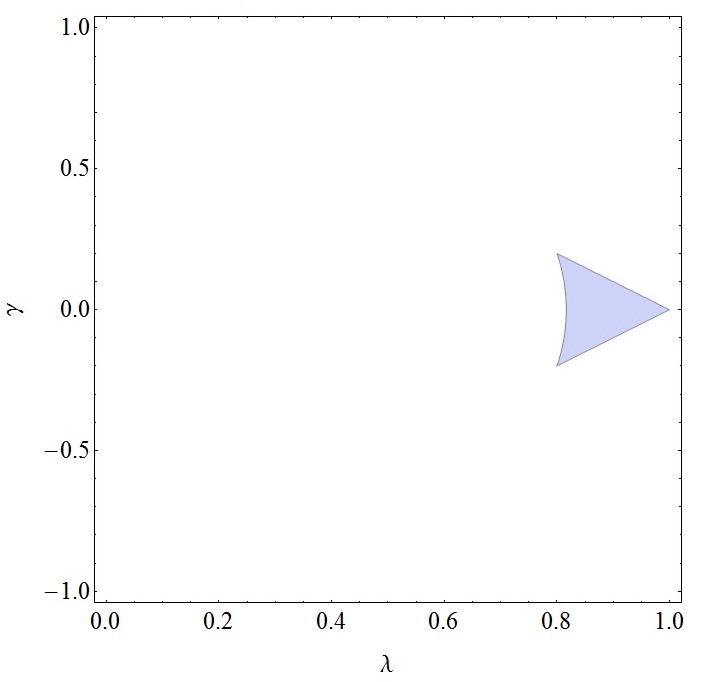}}
    \end{center}
    \caption{Blue shaded region indicates the range of $\lambda$ and $\gamma$ where QM violation of LGI given by Inequality (\ref{lgi}) is obtained with biased unsharp measurements.} \label{lgif}
  \end{figure}
Now consider the measurements corresponding to the operator $\sigma_z$ are biased unsharp. In this case the LHS of LGI given by Inequality (\ref{lgi}) is not independent of the state parameters. Let us choose $r=0$ and $\omega \Delta t = \frac{\pi}{6}$ (which gives maximum QM violation of the LGI (\ref{lgi}) as discussed before). With this the LHS of LGI (\ref{lgi}) is given by,
\begin{equation}
K_{LGI} = \frac{3}{2} \lambda^2 + \gamma^2.
\label{lgie}
\end{equation} 
From Eq.(\ref{lgie}) it is evident that $ K_{LGI} = \frac{3}{2}$ for projective measurements corresponding to the operator $\sigma_z$. Hence, in this case unsharpness of the measurements (quantified by ``$1-\lambda$") reduces the QM violation of LGI as $0 \leq \lambda \leq 1$. On the other hand, for a given unsharpness, incorporating any non-zero biasedness $\gamma$ in the measurements increases the QM violation of LGI.

In Fig.(\ref{lgif}) we show the region of $\lambda$ and $\gamma$ in which the QM violation of LGI given by inequality (\ref{lgi}) is obtained (constrained with the condition $|\lambda|+|\gamma| \leq 1$) using the expression of $K_{LGI}$ mentioned in Eq.(\ref{lgie}). From this Figure it is clear that introduction of biasedness parameter $\gamma$ enables to get QM violation of LGI given by inequality (\ref{lgi}) even for $\lambda \leq \sqrt{2/3}$.

We now consider the effect of biasedness parameter over that of sharpness parameter in the context QM violation of WLGI \cite{saha}, which is another necessary condition of MR.  We have chosen the following two partcular forms of WLGI giving the maximum QM violation \cite{saha}:
\begin{equation}
\label{wlgi1}
K_{WLGI1} = p ( Q_2 =+, Q_3=+) - p (Q_1 = -, Q_2 = +)
 - p (Q_1 = +, Q_3 = +) \leq 0,
\end{equation} 
and
\begin{equation}
\label{wlgi2}
K_{WLGI2} = p ( Q_2 =-, Q_3=-) - p (Q_1 = +, Q_2 = -) 
 - p (Q_1 = -, Q_3 = -) \leq 0.
\end{equation}
The maximum values of $K_{WLGI1}$ and $K_{WLGI2}$ mentioned in Inequalities (\ref{wlgi1}) and (\ref{wlgi2}) with projective measurements ($\lambda = 1$, $\gamma=0$) corresponding to the operator $\sigma_z$ at instances $t_1$, $t_2$ and $t_3$ are equal to $0.50$.

The maximum value of $K_{WLGI1}$ mentioned in Inequality (\ref{wlgi1}) occurs for $r =1$, $\theta=\frac{\pi}{3}$, $\phi=\frac{\pi}{2}$ and $\omega \Delta t = 0.56$. It has been shown that if the measurements corresponding to the operator $\sigma_z$ are unbiased unsharp ($0 \leq \lambda \leq 1$, $\gamma=0$), then WLGI given by inequality (\ref{wlgi1}) is not violated for $\lambda \leq 0.69$ \cite{saha}.

Now consider that the measurements corresponding to the operator $\sigma_z$ are biased unsharp. Let us take $r =1$, $\theta=\frac{\pi}{3}$, $\phi=\frac{\pi}{2}$ and $\omega \Delta t =0.56$ (which gives maximum QM violation of WLGI (\ref{wlgi1})). In this case the LHS of WLGI (\ref{wlgi1}) is given by,

\begin{equation}
\label{wlgie}
K_{WLGI1} = -\frac{1}{4}+ (\frac{\gamma}{2} + 0.61 \lambda)^2 - 0.12 \lambda \gamma
 + 0.19 \lambda (2 -  \sqrt{1-\gamma-\lambda} \sqrt{1-\gamma+\lambda}
-  \sqrt{1+\gamma-\lambda} \sqrt{1+\gamma+\lambda}).
\end{equation}

From Eq.(\ref{wlgie}) it can be checked that unsharpness of the measurements (quantified by ``$1-\lambda$") reduces the QM violation of WLGI (\ref{wlgi1}). On the other hand, for a given unsharpness, introducing any non-zero positive value of biasedness $\gamma$ increases the QM violation of WLGI (\ref{wlgi1}). In a similar way, it can be checked that, for a given unsharpness, introduction of any non-zero negative value of biasedness $\gamma$ increases the QM violation of WLGI (\ref{wlgi2}). 

In Fig.(\ref{wlgif}) we show the region of $\lambda$ and $\gamma$ in which the QM violation of WLGI given by Inequality (\ref{wlgi1}) is obtained (constrained with the condition $|\lambda|+|\gamma| \leq 1$) for $r =1$, $\theta=\frac{\pi}{3}$, $\phi=\frac{\pi}{2}$ and $\omega \Delta t =0.56$ (which gives maximum QM violation of WLGI (\ref{wlgi1})) using the expression of $K_{WLGI1}$ mentioned in Eq.(\ref{wlgie}). From this Figure it is clear that introduction of biasedness parameter $\gamma$ enables to get QM violation of WLGI given by Inequality (\ref{wlgi1}) even for $\lambda \leq 0.69$.
\begin{figure}[t!]
    \begin{center}
    \resizebox{9cm}{9cm}{\includegraphics{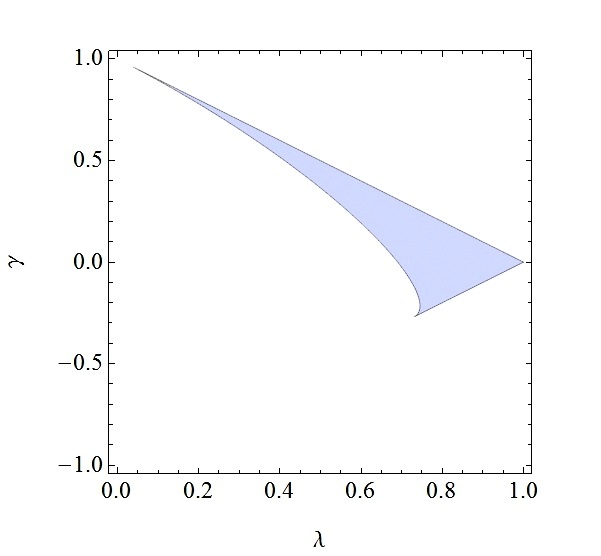}}
    \end{center}
    \caption{Blue shaded region indicates the range of $\lambda$ and $\gamma$ where QM violation of WLGI given by Inequality (\ref{wlgi1}) is obtained with biased unsharp measurements.} \label{wlgif}
  \end{figure} 

Note that in Fig.(\ref{wlgif}) the region of $\lambda$ and $\gamma$, in which QM violation of WLGI given by Inequality (\ref{wlgi1}) is enabled, is not symmetric with respect to $\gamma=0$. This is due to the particular choice of WLGI. QM violation of WLGI given by Inequality (\ref{wlgi1}) is enabled for a larger region of $\lambda$ and $\gamma$ when $\gamma$ is positive. Similarly, the region of $\lambda$ and $\gamma$, in which QM violation of WLGI given by Inequality (\ref{wlgi2}) is enabled, can be shown. This region for WLGI given by Inequality (\ref{wlgi2}) is the mirror image of the region for WLGI given by Inequality (\ref{wlgi1}) with respect to $\gamma=0$. Hence, the QM violation of WLGI given by Inequality (\ref{wlgi2}) is enabled for a larger region of $\lambda$ and $\gamma$ when $\gamma$ is negative.
 
We now consider the effect of biasedness parameter $\gamma$ over that of sharpness parameter $\lambda$ in the context QM violation of NSIT \cite{nsit}, which is another necessary condition of MR. We have chosen the following partcular form of NSIT giving the maximum QM violation \cite{saha}:
\begin{equation}
\label{nsit}
K_{NSIT} =  p(Q_1=-,Q_2=+) + p( Q_1=+, Q_2 =+) 
 - p(Q_2 = +) = 0.
\end{equation} 
The maximum value of $K_{NSIT}$ mentioned in Eq.(\ref{nsit}) with projective measurements ($\lambda = 1$, $\gamma=0$) corresponding to the operator $\sigma_z$ at instances $t_1$, $t_2$ is $0.5$. This maximum value of $K_{NSIT}$ occurs for $r =1$, $\theta=\frac{\pi}{2}$, $\phi=\frac{3 \pi}{2}$ and $\omega \Delta t = \frac{\pi}{4}$. 

\begin{figure}[t!]
    \begin{center}
     \resizebox{9cm}{9cm}{\includegraphics{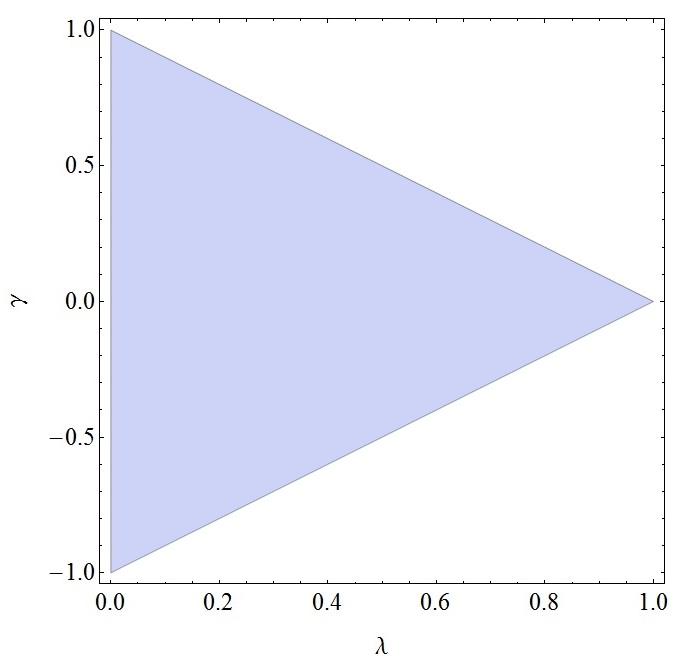}}
    \end{center}
    \caption{Blue shaded region indicates the range of $\lambda$ and $\gamma$ where QM violation of NSIT given by Eq.(\ref{nsit}) is obtained with biased unsharp measurements.} \label{nsitf}
  \end{figure}

It has been shown that if the measurements corresponding to the operator $\sigma_z$ are unbiased unsharp ($0 \leq \lambda \leq 1$ and $\gamma=0$), then NSIT given by Eq.(\ref{nsit}) is violated for $\lambda > 0$ \cite{saha}.

Now consider that the measurements corresponding to the operator $\sigma_z$ are biased unsharp. Let us take $r =1$, $\theta=\frac{\pi}{2}$, $\phi=\frac{3 \pi}{2}$ and $\omega \Delta t = \frac{\pi}{4}$ (which gives maximum QM violation of NSIT given by Eq.(\ref{nsit})). In this case the LHS of  NSIT (\ref{nsit}) is given by,
\begin{equation}
\label{nsie}
K_{NSIT} = \frac{1}{4} \lambda (2 - \sqrt{1-\gamma-\lambda} \sqrt{1-\gamma+\lambda}
 - \sqrt{1+\gamma-\lambda} \sqrt{1+\gamma+\lambda}) .
\end{equation} 
From Eq.(\ref{nsie}) it can be checked that unsharpness of the measurements (quantified by ``$1-\lambda$") reduces the QM violation of NSIT (\ref{nsit}). On the other hand, for a given unsharpness, introducing any nonzero biasedness $\gamma$ increases the QM violation of NSIT (\ref{nsit}).

In Fig.(\ref{nsitf}) we show the region of $\lambda$ and $\gamma$ in which the QM violation of NSIT given by Eq.(\ref{nsit}) is obtained (constrained with the condition $|\lambda|+|\gamma| \leq 1$) for $r =1$, $\theta=\frac{\pi}{2}$, $\phi=\frac{3 \pi}{2}$ and $\omega \Delta t = \frac{\pi}{4}$ (which gives maximum QM violation of  NSIT (\ref{nsit}))  using the expression of $K_{NSIT}$ mentioned in Eq.(\ref{nsie}).  From this Figure it is clear that introduction of biasedness parameter $\gamma$ makes no difference in the minimum value of $\lambda$ above which  NSIT (\ref{nsit}) is violated. Since, for $\gamma = 0$, NSIT is violated for any $\lambda > 0$, there is no scope to decrease the minimum value of $\lambda$ above which NSIT is violated by introducing nonzero $\gamma$.

It can, therefore, be said that testing of MR is more robust than testing of LR in demonstrating the trade-off between biasedness and unsharpness of measurements.

In practical scenario measurements are not ideal projective. Non-idealness present in measurements can be modeled theoretically by constructing different POVMs. Unbiased unsharp measurement ($0 \leq \lambda \leq 1$ and $\gamma=0$) is one such POVM. For dichotomic spin measurements, unsharpness is linked with the non-idealness due to non-zero overlap between non-orthgonal probe states \cite{mal1}. Biased unsharp measurement is another form of POVM where biasedness characterizes non-idealness in the alignment of Stern-Gerlach apparatus or the deviation from the Gaussian nature of spatial wave-packet of the incident spin-$\frac{1}{2}$ particles \cite{mal}. 

QM violation of any necessary condition of MR can be considered as a signature of non-classicality of the system under consideration. The aforementioned results indicate that non-idealness as captured by unsharpness of the measurements reduces the non-classical feature in the form of QM violations of MR for qubit systems. On the other hand, for a given unsharpness, non-idealness as captured by biasedness of the measurements increases the non-classical feature as evidenced by the QM violations of MR. Moreover, the presence of non-zero biasedness in the measurements enables one to get QM violation of LGI or WLGI for a certain amount of unsharpness, even when this violation is not observed for that certain unsharpness in the absence of non-zero biasedness. 

Hence, it can be concluded that, in the context of QM violations of MR, unsharpness of measurements reduces non-classicality of the qubit system under consideration. Surprisingly, one can diminish this effect of non-idealness of measurements associated with unsharpness by incorporating another form of non-idealness, say, biasedness in the measurements. It is to be noted that the proper physical reason behind the trade-off between unsharpness and biasedness of measurement is not very clear sofar.

\section{Discussions and conclusion}
In quantum theory all possible measurements are represented by POVM formalism. The set of all dichotomic POVMs are fully characterized by two real parameters, say, sharpness parameter and biasedness parameter. These two parameters can be linked with different aspects of the experimental set-up in the laboratory. Sharpness parameter is linked with the precision of measurements while biasedness is related to the error in the alignment of Stern-Gerlach apparatus or deviation from Gaussian nature of spatial wave-packet of the incident particles. In this work we have shown that there is a trade-off between sharpness parameter and biasedness parameter characterizing arbitrary dichotomic POVM in the context of probing QM violations of local-realist and macro-realist inequalities. 

In case of spatial correlations we have investigated the effect of the biasedness parameter on the minimum values of the sharpness parameter above which the QM violation of CHSH inequality persists. We have also derived the necessary and sufficient criteria for violating the CHSH inequality with most general dichotomic POVMs at one side and projective measurements at another side. We have found that if a state does not exhibit QM violation of the CHSH inequality with projective measurements at both sides, then that cannot be enabled with dichotomic POVMs at one side. This result is consistent with a previous result which states that, for two dichotomic measurements per party, projective measurements are sufficient for the maximal violation of CHSH inequality considering optimal quantum strategy for binary XOR game \cite{clave}.

In case of temporal correlations we have demonstrated that, by introducing non-zero biasedness in the measurements, the minimum value of the sharpness parameter above which QM violation of LGI or WLGI persists can be reduced even below the minimum value of sharpness required for probing QM violation of LGI or WLGI in the absence of non-zero biasedness. However, the minimum value of the sharpness parameter over which QM violation of NSIT persists remains the same by introducing any non-zero biasedness parameter in the measurements. This is because any non-zero value of sharpness parameter enables one to obtain QM violation of NSIT in the absence of biasedness. Hence, there is no scope to reduce the minimum value of the sharpness parameter above which QM violation of NSIT persists by introducing any non-zero biasedness.  We can, therefore, conclude that biasedness of measurements can counter the effect of unsharpness of measurements in probing quantumness of the system under consideration by demonstrating QM violations of local-realist and macro-realist inequalities and this kind of trade-off is more robust in the latter case. 

Note that the above trade-off is counter-intuitive and it is non-trivial to underpin the physical reason behind this aforementioned trade-off. Finding out why biasedness of measurements helps in detecting quantum mechanical features for a fixed unsharpness may be interesting for future studies.\\\\

\begin{center}
\textbf{\textit{Acknowledgements:}}
\end{center}
DH and SM acknowledge fruitful discussions with Prof. P. Busch during his visit at SNBNCBS, Kolkata, in the year of 2015. DD acknowledges the financial support from University Grants Commission (UGC), Government of India. We also thank anonymous referees for precious comments which facilitates in improving the earlier version of the draft.

\end{document}